\documentclass[a4paper,12pt,twoside]{article}
\title{{Optimal Selling Time of a Stock under Capital Gains Taxes}}
\author{C. K\"uhn\footnote{Institut f\"ur Mathematik, Goethe-Universit\"at Frankfurt, D-60054 Frankfurt a.M., Germany, email: 
\{ckuehn,ulbricht\}@math.uni-frankfurt.de} \quad B. A. Surya\footnote{School of Business and Management, Bandung Institute of Technology, Bandung, West Java, Indonesia, email: 
budhi.surya@sbm-itb.ac.id. B. A. Surya was visiting Goethe University  
when part of this work was carried out. Financial support 
of the Deutscher Akademischer Austausch Dienst~(DAAD) is gratefully 
acknowledged.} \quad B. Ulbricht\footnotemark[1]{}}
\usepackage{stmaryrd}
\usepackage{color}
\usepackage{epsfig}
\usepackage{amsfonts}
\usepackage{amsmath}
\usepackage{amssymb}
\usepackage{subfigure}
\usepackage{graphicx}
\usepackage{theorem}
\usepackage{mathrsfs}
\usepackage{listings}

\rm 

\newtheorem{theorem}{Theorem}[section]
\newtheorem{theo}[theorem]{Theorem}
\newtheorem{lem}[theorem]{Lemma}

\newtheorem{prop}[theorem]{Proposition}

{\theorembodyfont{\normalfont}

\newtheorem{Rem}[theorem]{Remark}

}
\renewcommand{\tilde}{\widetilde}

\newcommand{\exit}{{\mbox{\, \vspace{3mm}}} \hfill\mbox{$\square$}}

\def\bbr{{\Bbb R}}   

\def\wt{\widetilde}

\newcommand{\beao}{\begin{eqnarray*}}
\newcommand{\eeao}{\end{eqnarray*}\noindent}

\newcommand{\beam}{\begin{eqnarray}}
\newcommand{\eeam}{\end{eqnarray}\noindent} 

\newcommand{\eps}{\varepsilon}

\numberwithin{equation}{section}

\begin{document}
\maketitle \pagestyle{myheadings} \markboth{C. K\"uhn, B. A. Surya,
and B. Ulbricht} {Optimal Selling Time of a Stock under Capital Gains Taxes}
\begin{abstract}
We investigate the impact of capital gains taxes on optimal investment decisions in a quite simple model.
Namely, we consider a risk neutral investor who owns one risky stock from which she assumes that it has a lower expected return than the riskless 
bank account and determine the optimal stopping time at which she sells the stock to invest the proceeds in the bank account up to the maturity date. 
In the case of linear taxes and a positive riskless interest rate,
the problem is nontrivial because at the selling time the investor has to realize book profits which triggers tax payments. 
We derive a boundary that is continuous and increasing in time and decreasing in the volatility of the stock such that the investor sells the stock at the first time 
its price is smaller or equal to this boundary.

\medskip



\end{abstract}


\section{Introduction}

In practice, capital gains taxes are undoubtedly one of the most relevant market frictions.
In most countries, an important feature of capital gains taxes
is the rule that profits are taxed when the asset is liquidated, i.e., the gain is realized, and not when gains actually occur.
Thus, even in the simplest case of a {\em linear} taxing rule (that we consider in the current paper),
there is a nontrivial interrelation between creating trading gains and tax liabilities by dynamic investment strategies.
The investor can influence the timing of the tax payments, i.e.,
she holds a deferral option. In the case of a positive riskless interest rate, there is some incentive to realize profits as late as possible, but this can be
at odds with portfolio regroupings in order to earn higher returns before taxes.

Solving a portfolio optimization problem with taxes allowing for arbitrary continuous time trading strategies is a rather daunting task, especially
for the so-called exact tax basis and the first-in-first-out priority rule.  Namely, shares having the same price but being purchased at different times possess, in general, 
different book profits, and hence, their liquidation triggers different tax payments. Allowing the investor to choose which of the shares is relevant for
taxation is called exact tax basis. The book profits of the shares in the portfolio become an infinite-dimensional state variable (cf. Jouini, Koehl, and 
Touzi~\cite{jouini1, jouini2} for the first-in-first-out priority rule and \cite{kuehn.ulbricht} for the exact tax basis). 

Dybvig and Koo~\cite{dyb1} and DeMiguel and Uppal~\cite{dem1} model the exact tax basis in discrete time and relate the portfolio optimization problem to nonlinear programming.
Jouini, Koehl, and Touzi~\cite{jouini2,jouini1} consider the first-in-first-out priority rule with one nondecreasing asset price, but with a quite general tax code,
and derive first-order conditions for the optimal consumption
problem. The problem consists of injecting cash from the income stream into the single asset and withdrawing it for consumption. 
Ben Tahar, Soner, and Touzi~\cite{bentaharsonertouzi10,bentaharsonertouzi07} solve the Merton problem
with proportional transaction costs and a tax based on the average of past purchasing prices.
This approach has the advantage that the optimization problem is Markovian with the only one-dimensional tax basis
as additional state variable. 
Cadenillas and Pliska~\cite{cadenillas1999optimal} and Buescu, Cadenillas, and Pliska~\cite{buescu2007note} maximize the long-run growth rate of investor's wealth
in a model with taxes and transaction costs. Here, after each portfolio regrouping, the investor has to pay capital gains taxes for her total portfolio. 

In practise, an investor is usually interested in much simpler optimization problems. Therefore, we want to analyze a typical and analytically quite tractable investment decision problem 
to determine exemplarily the impact of capital gains taxes and to see how model parameters, as the volatility of the stock, enter into the solution. Often, the investor wants to maximize her
trading profits within a certain finite period of time by exchanging
one asset for another one only once. This means that she has to solve an optimal stopping problem.
In this simple setting, different tax basises, as, e.g., the exact tax basis, the average tax basis, or the first-in-first-out priority rule, coincide.
To investigate the impact of taxes on investment decisions, 
we look at an investor owning an asset which she would sell immediately to buy another one if she was {\em not} subject to taxation. 
%
%
%
Under risk neutrality, this just means that the
asset the investor holds at the beginning has a lower expected return than the alternative asset. Then, we investigate to what extent she is prevented from this transaction
by the obligation to pay taxes at the time she liquidates the first asset. The price of the first asset is modeled as {\em stochastic} process in the Black-Scholes market
to see the influence of the volatility on the
deferral option the tax payer holds. We prove the plausible supposition that the possibility to time the tax payments is more worthwhile for holders of more volatile assets
and, consequently, the risky asset is sold later (see Proposition~\ref{2.12.2014.1}).
We assume that the second asset is then kept in the portfolio up to maturity. Thus, it is no essential restriction to model it as a riskless bank account.

In contrast to the classic model of  Constantinides~\cite{constantinides.1983}, we do {\em not} assume that the investor can both defer the tax payments and divest the stock from 
her portfolio at the same time by trading in a market for short sell contracts (see Subsection~4.1 of \cite{constantinides.1983}).
To our mind, it is mainly an interesting gedankenexperiment to shorten the stock, instead of selling it, in order to defer tax payments, but under real-world 
tax legislation, it is no option for private investors. By assuming the existence of such a market for short sell contracts, 
Constantinides can price the timing option in the Black-Scholes model by no-arbitrage arguments and without solving a {\em free}-boundary problem.
Another essential difference to \cite{constantinides.1983} is that we have a deterministic finite time horizon, whereas in
\cite{constantinides.1983} liquidation is forced at independent Poisson times. Thus, as in problems with infinite time horizon, in the latter article, one gets rid of `time' 
as a state variable.  

In this paper, some standard techniques from the theory of optimal stopping are used, especially an approach that turns the problem with a terminal payoff to one with
a {\em running} payoff, see
Peskir and Shiryaev~\cite{peskir}. The objective function is much simpler than in other recent papers on the optimal selling time of a stock without taxes, as
Shiryaev, Xu, and Zhou~\cite{Shiryaev}, where the stock is sold at the stopping time which maximizes the
expected ratio between the stock price and its maximum over the entire horizon. Du Toit and Peskir \cite{duToit} complement this by determining a stopping time that minimizes 
the expected ratio of the ultimate maximum and the current stock price. Dai and Zhong~\cite{Dai} consider a
similar problem in which the average stock price is used as
reference. In addition to the above selling problems, in their
recent work, Baurdoux et al. \cite{Baurdoux} discuss a `buy low and
sell high' problem as sequential optimal stopping of a Brownian bridge
modeling stock pinning. This is a phenomenon where a stock price tends to end up in the vicinity
of the strike of its option near its expiry, see \cite{Avellaneda} for a detailed explanation.

The paper is organized as follows. In Section~\ref{sec:probform},
we formulate the optimal stopping problem and present its solution (Theorem~\ref{theo:main2}), which is, accompanied by
Proposition~\ref{2.12.2014.1}, the main result of the paper. Afterwards, the results are related to other contributions in the literature,
and the tax-timing option is valued (Remark~\ref{30.12.2014.1}).
Section~\ref{sec:method} introduces the applied method to solve the problem and prepares the proofs which
are given in  Section~\ref{sec:proof}.

\section{Formulation of the Stopping Problem and Main Results}\label{sec:probform}
Consider a filtered probability space $(\Omega,\mathcal{F},(\mathcal{F}_t)_{t\in[0,T]},\mathbb{P})$, $T\in\bbr_+$,
generated by a one-dimensional standard Brownian motion~$(B_t)_{t\in[0,T]}$. 
The investment oppurtunities consist of a bank account with
continuously compounded fixed interest rate $r>0$ and a stock whose price process~$(X_t)_{t\in[0,T]}$ solves the SDE
\begin{equation}\label{eq:stock}
dX_t=\mu X_t dt + \sigma X_t dB_t, \quad t\geq 0,
\end{equation}
with $X_0=x>0$, $\mu\in\bbr$ and $\sigma>0$. 
At time~$0$, the investor holds one risky stock. It was purchased sometime in the past at price~$P_0>0$.
This means that already at time~$0$, the stock possesses the book profit~$X_0-P_0$. The economically interesting case is $X_0-P_0>0$, but this need not be assumed.
The investor can sell the stock 
at any time up to the end of the time horizon~$T$. At the selling time~$t$, the investor has to pay the capital gains taxes~$\alpha(X_t-P_0)$, 
where $\alpha\in[0,1)$ is the given tax rate, i.e., if $X_t<P_0$, the investor gets a tax credit. Then, the remaining wealth~$X_t - \alpha(X_t-P_0) = (1-\alpha)X_t + \alpha P_0$ 
is invested in the riskless bank account. At maturity~$T$, the portfolio is liquidated anyway. 
As the bank account pays a continuous compounded interest rate, we assume that taxes
also charge the account continuously. This corresponds to the taxation of a continuous dividend flow and leads to the after-tax interest rate~$(1-\alpha)r$.
Thus, the investor's wealth at maturity when selling the stock at time $t\in[0,T]$ at price~$\wt{x}$ is
\begin{equation}\label{eq:payoff}
G(t,\wt{x})=\big[(1-\alpha)\wt{x}+\alpha P_0\big]e^{r(1-\alpha)(T-t)},\quad (t,\wt{x})\in[0,T]\times\mathbb{R}_+.
\end{equation}
Maximizing investor's expected wealth at maturity leads to the optimal stopping problem
\begin{equation}\label{eq:osp_at0}
V(x):=\sup_{\tau\in\mathcal{T}_{[0,T]}}\mathbb{E}\left[G(\tau,X_{\tau})\right],
\end{equation}
where $X_0=x$, and by $\mathcal{T}_{[0,T]}$ we denote the set of
$(\mathcal{F}_{s})_{s\in[0,T]}-$stopping times taking values in $[0,T]$. The assumption that the second investment opportunity is a riskless bank account 
rather than another risky asset makes the payoff function
a bit more tractable and is, given that the investor is risk neutral and cannot change her position again before $T$, not very restrictive.

Because of the strong Markov property of $(X_t)_{t\in[0,T]}$, we define the value function associated with problem (\ref{eq:osp_at0}) by 
\begin{equation}\label{eq:osp}
V(t,x):=\sup_{\tau\in\mathcal{T}_{[0,T-t]}}\mathbb{E}
\big[G(t+\tau,X^{t,x}_{t+\tau})\big],
\end{equation}
where $(X^{t,x}_s)_{s\in[t,T]}$ is the unique solution of (\ref{eq:stock}) with initial condition $X_t^{t,x}=x$ and $\mathcal{T}_{[0,T-t]}$ denotes the
set of $(\mathcal{F}_{t+s})_{s\in[0,T-t]}-$stopping times taking values in
$[0,T-t]$. Note that $V(x)=V(0,x)$ as $X^{0,x}=X$ with $X_0=x$.
By setting $\tau=0$ in (\ref{eq:osp}), it is clear that
\begin{equation}\label{eq:majorant}
V(t,x)\geq  G(t,x) \quad \textrm{for} \;\;
(t,x)\in[0,T]\times\mathbb{R}_+.
\end{equation}
In addition, we have that
\begin{equation}\label{eq:BC}
V(T,x)=G(T,x)\quad \mbox{for}\ x\in\mathbb{R}_+.
\end{equation}
The continuation region is defined by
$${\mathcal{C}}:=\{(t,x)\in[0,T)\times\mathbb{R}_+\ |\ V(t,x)>G(t,x)\}$$
and the stopping region by
$${\mathcal{S}}:=\{(t,x)\in[0,T)\times\mathbb{R}_+\ |\ V(t,x)=G(t,x)\}.$$
\begin{prop}\label{prop:continue}
$(t,x)\mapsto V(t,x)$ is continuous on
$[0,T]\times\mathbb{R}_+$. In addition, for any $(t,x)\in[0,T]\times \bbr_+$, the stopping time
\begin{equation}\label{eq:ost}
\tau_{t,x}:=\inf\{s\in[0,T-t)\ |\ (t+s,X^{t,x}_{t+s})\in \mathcal{S}\}\wedge (T-t)
\end{equation}
maximizes (\ref{eq:osp}).
\end{prop}
Of course, the proof follows more or less directly from the standard theory. Thus, it is only briefly sketched in Section~\ref{sec:proof}.
The following theorem states the main results of this paper and is also proved in Section~\ref{sec:proof}.
\begin{theo}\label{theo:main2}
Consider problem (\ref{eq:osp}) with stopping region
$\mathcal{S}$. Let $r\geq0$ and $\alpha\in[0,1)$. Then, 
\begin{itemize}
\item[(a)] there exists a continuous, increasing boundary $b:[0,T)\rightarrow\mathbb{R}_+$
such that the stopping region is given by
\begin{equation}\label{eq:boundary}
\mathcal{S} =\left\{
\begin{array}{ll}
[0,T)\times \mathbb{R}_+ & \textrm{if}\text{ } \mu\leq (1-\alpha)r \\
\{(t,x)\in[0,T)\times\mathbb{R}_+\ |\ x\leq b(t)\} &
\textrm{if}~ \mu>(1-\alpha)r,\\
\end{array}
\right.
\end{equation}
where for all $t\in[0,T)$, the equivalence $\alpha>0 \Leftrightarrow b(t)>0$ holds. \\
The boundary satisfies the terminal condition 
\begin{align*}
\lim_{t\uparrow T} b(t)=\frac{r\alpha P_0}{\mu-r(1-\alpha)}=:f.
\end{align*}
\item[(b)] if $\alpha>0$ and $\mu>(1-\alpha)r$, the value function satisfies the smooth fit condition at the boundary, i.e., 
\begin{equation}\label{eq:smoothfit}
\partial_xV(t,x)=\partial_xG(t,x)=(1-\alpha)e^{r(1-\alpha)(T-t)}\quad \textrm{at $x=b(t)$.}
\end{equation}
\end{itemize}
\end{theo}
\begin{figure}[ht!]
\begin{center}
\epsfig{file=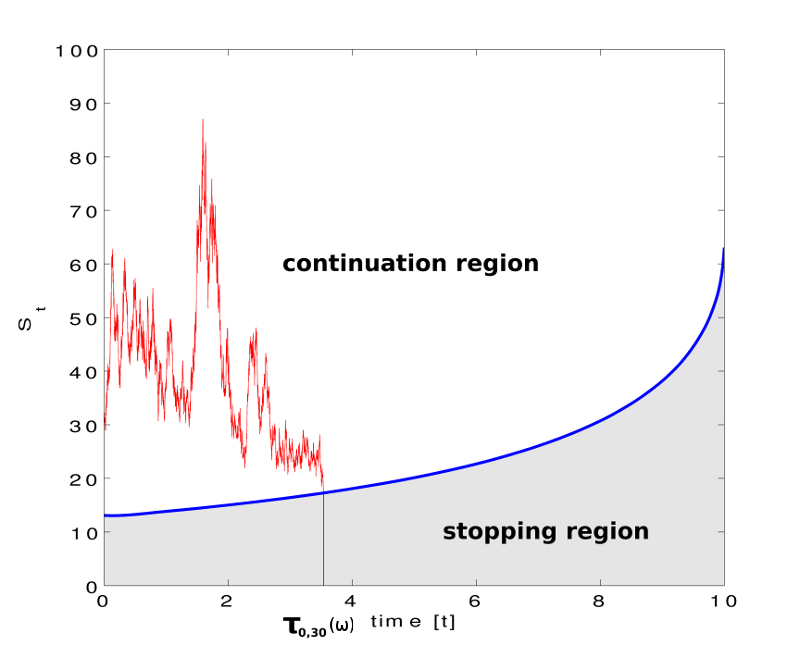, width=0.685\textwidth}
\caption{Plot of the stopping boundary in the optimization problem~(\ref{eq:osp_at0}) with time horizon~$T=3$ years. 
The other parameters are $\alpha=0.3$, $\mu=0.026$, $r=0.03$, $\sigma=0.25$, $P_0=100$, and $x=180$. The boundary lies far above the purchasing price~$P_0$, i.e.,
the stock is sold with positive book profits.}\label{fig:stoppingboundary}
\end{center}
\end{figure}

\begin{Rem}
We are primarily interested in the case that $\mu<r$, i.e., an investor who is not subject to taxation would immediatelly sell the stock.
Then, Theorem~\ref{theo:main2} tells us that the solution is nontrivial if $\mu\in ((1-\alpha)r,r)$. Choosing the midpoint of the interval for $\mu$ and reasonable
values  for $\sigma$, $T$, and $\alpha$, numerical calculations show that the boundary is far above the purchasing price~$P_0$, see Figure~\ref{fig:stoppingboundary}. 
This means that the investor always sells the stock with positive book profits. Then, it is evident that there is no 
incentive to buy the stock back at any later time and to repeat the game (namely, a renewed investment starts with zero book profit). This justifies the modeling of the decision 
problem as a simple optimal stopping problem.

On the other hand, if the stock is sold with negative book profits, the modeling is justified when wash-sales are forbidden, as, e.g., to some extent in the U.S. 
This means that the investor may sell the stock to realize the trading loss, but then she is not allowed to buy back the stock and has to take the other investment opportunity. 
Under the ban on wash sales, the investor may switch to the bank account even in the case~$\mu\ge r$,
just in order to realize losses prematurely, which leads to a nontrivial solution of the stopping problem (see Theorem~\ref{theo:main2}).
\end{Rem} 

\begin{prop}\label{2.12.2014.1}
(i) The value function is increasing in the volatility of the stock, i.e., $V^{\sigma_1}(t,x)\le V^{\sigma_2}(t,x)$ for all $0\le \sigma_1\le \sigma_2$, $t\in[0,T]$, and $x\in\bbr_+$. 
Consequently, $b^{\sigma_2}(t)\le b^{\sigma_1}(t)$ for all $\mu>(1-\alpha)r$.\\

(ii) For $\sigma=0$, the exercise boundary reads
\beam\label{17.12.2014.1}
b(t)=\frac{\alpha P_0\left(e^{r(1-\alpha)(T-t)}-1\right)}{(1-\alpha)\left(e^{\mu(T-t)}-e^{r(1-\alpha)(T-t)}\right)},
\eeam
and the optimal stopping time~(\ref{eq:ost}) is given by
\begin{equation*}
\tau_{t,x} =\left\{
\begin{array}{ll}
0 & \textrm{if}\text{ } x\le b(t)\\
T-t & \textrm{if }x>b(t).\\
\end{array}
\right.
\end{equation*}
\end{prop}
Proposition~\ref{2.12.2014.1} that is proved in Section~\ref{sec:proof} makes sense from an economic point of view. For a more volatile asset, the option to time the tax payments has a
higher value for investors. {\em This means that capital gains taxes can even motivate investors to take more risk.} 
Of course, the extent of this effect depends on the riskless interest rate~$r$ and vanishes for $r=0$. 

Seifried~\cite{seifried.2010} solves the utility maximization problem for terminal wealth with $r=0$ (i.e., it can be assumed that taxes are paid at maturity), 
but there are {\em no tax credits}. In the model of \cite{seifried.2010}, there can appear two opposite effects. Roughly speaking, if the drift of the stock is low, the tax may  
prevent the investor to buy the risky stock at all because, without negative taxes on losses, the expected after-tax gain becomes negative (literally, this 
only holds for buy-and-hold strategies 
in the stock). On the other hand, if 
the expected return is high enough, the investor may buy even more risky stocks than in the same situation without taxes. To make the latter plausible,
consider a one-period binary model for the stock and a utility function that is linear around the initial wealth and satiable at some higher level of wealth.
For the optimal stock position, the investor's terminal wealth coincides with the saturation point if the stock price goes up. This means that, in the case of taxes, 
the investor buys even more risky stocks to offset the part of the gains that she has to pay to the government.  
\cite{seifried.2010} derives that for the Black-Scholes model with CRRA investors and realistic tax rates, the overall strategy effect of taxes is negligible (see Figure~8 therein).

\begin{Rem}[Value of the tax-timing option]\label{30.12.2014.1}
Assume that in our model, the book profit~$x-P_0$ is taxed at time $0$, and later gains on the stock are taxed immediately when they occur. Thereby, we assume that
tax payments are financed by reducing the stock position, and tax rebates are reinvested in stocks. Then, the wealth in stocks satisfies the SDE 
\beao
dX_t = (1-\alpha)X_t\left(\mu\,dt + \sigma\,dB_t\right),\quad t\ge 0\quad\mbox{with\ }X_0=(1-\alpha)x + \alpha P_0
\eeao
and thus $E[X_T]=\left[(1-\alpha)x + \alpha P_0\right]e^{\mu(1-\alpha)T}$. The optimal stopping problem degenerates: for $\mu\le r$, it is optimal to sell the stock at time~$0$, and for  
$\mu\ge r$, it is optimal to sell at time~$T$.

Thus, one may interpret
\beao
e^{-r(1-\alpha)T}\left(V(0,x)-\left[(1-\alpha)x + \alpha P_0\right]e^{(1-\alpha)\max(\mu,r) T}\right) 
\eeao
as the time~$0$ value of the {\em timing option} of the stock holder, i.e., the value of the right of the investor to influence the 
timing of the tax payments. By Proposition~\ref{2.12.2014.1}, this value is increasing in the volatility of the stock. 
This is in line with the results in Constantinides~\cite{constantinides.1983}, where in a complete market model, including a 
market for short sell contracts, the price of the timing option, of course differently defined (see Equation~(21) therein), is also increasing in the volatility of the stock 
(see Table~III). 
\end{Rem}

\section{Method of solution}\label{sec:method}
To solve problem (\ref{eq:osp}), we make use of a standard method in the optimal stopping theory, see, e.g., Peskir and Shiryaev~\cite{peskir}, where
the terminal payoff is turned into a running payoff. Namely, thanks to the smoothness property of the payoff
function $G$, we can apply It\^o's formula to obtain the
following decomposition of the payoff process $(G(t+s,X_{t+s}^{t,x}))_{s\geq 0}$:
\begin{equation}\label{eq:decom}
G\big(t+s,X_{t+s}^{t,x}\big)=G(t,x)+\int_0^sF\big(t+u,X_{t+u}^{t,x}\big)du+\mathcal{M}_s,
\end{equation}
where $\mathcal{M}_s=\sigma \int_0^s X_{t+u}^{t,x} \partial_xG(t+u,X_{t+u}^{t,x})dB_u$ is a
square integrable $(\mathcal{F}_{t+s})_{s\in[0,T-t]}-$ martingale with zero expectation, and $F(t,x)$ is given by
\begin{equation}\label{eq:payoffH}
\begin{split}
F(t,x)=&e^{r(T-t)(1-\alpha)}(1-\alpha)\left(-r\alpha P_0+x\big[\mu-r(1-\alpha)\big]\right).
\end{split}
\end{equation}
By (\ref{eq:decom}), $V$ from (\ref{eq:osp}) can be written as
\begin{equation}\label{eq:osp_new}
V(t,x)=G(t,x)+\sup_{\tau
\in\mathcal{T}_{[0,T-t]}}\mathbb{E}\left[\int_0^{\tau}F(t+u,X_{t+u}^{t,x})du\right].
\end{equation}
We have that 
\begin{align}\label{fgeq0}
F >(<)\;0~\mbox{on the set}~\{(t,x)\in[0,T]\times\mathbb{R}_+\ |\ x
>(<)\;f\},
\end{align} 
with
\begin{equation}\label{eq:funcf}
f:=\frac{r\alpha P_0}{\mu-r(1-\alpha)}.
\end{equation}
This means that the sign of $F(t,x)$ does not depend on $t$.
\begin{Rem}\label{rem:conti_region}
For $(t,x)$ with $x>f$, the stopping time $\tau:=\inf\big\{s\in[0,T-t]\ |\ X^{t,x}_{t+s}\le f\big\}\wedge (T-t)$ is strictly positive, and we conclude from
(\ref{eq:osp_new}) that $V(t,x)\geq G(t,x)+
\mathbb{E}\big\{\int_0^{\tau}F(t+u,X_{t+u}^{t,x})du\big\}>G(t,x),$ and thus $(t,x)\in\mathcal{C}$.
\end{Rem}

We have $f\geq0$ and
\begin{equation}\label{fxgeq0}
\mu > (1-\alpha)r\quad \implies \partial_x F(t,x)=e^{r(T-t)(1-\alpha)}(1-\alpha)\big(\mu-r(1-\alpha)\big)>0
\end{equation}
for all $(t,x)\in[0,T]\times\mathbb{R}_+$.
Furthermore,
\begin{equation}\label{ftleq0}
\partial_t F(t,x)=-r(1-\alpha)^2 e^{r(T-t)(1-\alpha)}\left(-r\alpha P_0+x\big[\mu-r(1-\alpha)\big]\right)>0
\end{equation}
for $x<f$, and
\begin{equation}\label{ftgeq0}
\partial_t F(t,x)<0\qquad\mbox{for $x>f$.}
\end{equation}

\section{Proofs}\label{sec:proof}
\subsection{Proof of Proposition \ref{prop:continue}}
For all $x\le y$, one has
\begin{eqnarray}
0&\leq &V(t,y)-V(t,x)\notag\\
&\leq&\sup_{\tau\in\mathcal{T}_{[0,T-t]}}\Big\{\mathbb{E}\left[G\big(t+\tau,X_{t+\tau}^{t,y}\big)\right]-\mathbb{E}\left[G\big(t+\tau,X_{t+\tau}^{t,x}\big)\right]
\Big\}\notag\\
&\leq & (1-\alpha)e^{(\mu+r)T} \mathbb{E}\left[\sup_{0\leq s\leq T-t}\exp\left(\sigma B_s - \frac12 \sigma^2 s\right)\right](y-x)\notag\\
& \le & C (y-x)
\label{eq:devVx}
\end{eqnarray}
for some constant $C\in\bbr_+$ that does not depend on $t,x,y$. Thus, to establish joint continuity in $(t,x)$, it remains to show that $t\mapsto V(t,x)$ is continuous. 
Let $s\le t$. One has
\beao
V(t,x) - V(s,x)
\le \sup_{\tau\in\mathcal{T}_{[0,T-t]}}\mathbb{E}\left[G(t+\tau,X_{t+\tau}^{t,x})-G(s+\tau,X_{t+\tau}^{t,x})\right]\le 0,
\eeao
where for the first inequality, we use that the processes $(X^{t,x}_{t+u})_{u\in[0,T-t]}$ and $(X^{s,x}_{s+u})_{u\in[0,T-t]}$ coincide in distribution.

To obtain an estimation in the other direction, one also has to find an upper bound for the increments of the payoff process~$u\mapsto G(u,X^{s,x}_u)$ between $s+T-t$ and $T$ because
the remaining time to maturity is smaller for the problem started in $t$. By the monotonicity of $x\mapsto F(u,x)$, one has
\beao
V(s,x) - V(t,x) & \le & \sup_{\tau\in\mathcal{T}_{[0,T-t]}}\mathbb{E}\left[G(s+\tau,X_{t+\tau}^{s,x})-G(t+\tau,X_{t+\tau}^{s,x})\right]\\
& & \quad  + \mathbb{E}\left[\int_{T-(t-s)}^T F\left(u,\sup_{v\in[s,T]} X_v^{s,x}\right)\vee 0\,du\right].
\eeao
The first term can be estimated by
\beao
\left(e^{-r(1-\alpha)s}-e^{-r(1-\alpha)t}\right) e^{r(1-\alpha)T}\mathbb{E}\left[(1-\alpha)e^{|\mu| T}\sup_{v\in[0,T]}x\exp\left(\sigma B_v - \frac12 \sigma^2 v\right) +\alpha P_0\right]
\eeao
and the second term by
\beao
(t-s)\left( C_1 + C_2 \mathbb{E}\left[\sup_{0\leq v\leq T}x\exp\left(\sigma B_v - \frac12 \sigma^2 v\right)\right]\right)
\eeao
for some constants $C_1,C_2\in\bbr_+$. Altogether, it follows that $t\mapsto V(t,x)$ is continuous for any fixed $x\in\mathbb{R}_+$.

Then, by Theorem~2.4 of \cite{peskir}, it follows that (\ref{eq:ost}) maximizes (\ref{eq:osp}).\exit 

\subsection{Proof of Theorem \ref{theo:main2}}
We distinguish three cases. The first two cases are trivial, whereas the third case is the interesting one, where we show the existence of a continuous, increasing, positive 
boundary such that the stock is sold at the first time its price is smaller or equal to this boundary.\hfill\vspace{0.5cm}\\
{\em{Case 1:}} $\mu\leq(1-\alpha)r$ \vspace{0.3cm}\\
>From (\ref{eq:payoffH}), we see that $F\leq0$. Then,
from (\ref{eq:osp_new}), it follows that $\tau_{t,x}=0$ for all $(t,x)\times[0,T]\times\bbr_+$, i.e., the investor sells the stock immediately and
invests the proceeds in the bank account.\hfill\vspace{0.5cm}\\
{\em{Case 2:}} {$\mu>(1-\alpha)r$} and $\alpha=0$ \vspace{0.3cm}\\
One has $F(t,x)>0$ for all $(t,x)\in[0,T]\times\mathbb{R}_+\setminus\{0\}$. Then, again from (\ref{eq:osp_new}), it follows that 
$\tau_{t,x}=T-t$ for all $(t,x)\in[0,T]\times\mathbb{R}_+\setminus\{0\}$, 
i.e., the investor never sells the stock prematurely and thus $\mathcal{S}=[0,T)\times\{0\}$.\hfill\vspace{0.5cm}\\
{\em{Case 3:}} {$\mu>(1-\alpha)r$} and $\alpha>0$ \vspace{0.3cm}\\
%
{\em Step 1:} Let us show that for every $t\in[0,T], x,y\in\bbr_+$ with $x\le y$, the implication
\beam\label{28.11.2014.1}
V(t,y)=G(t,y)\quad\implies\quad V(t,x)=G(t,x)
\eeam
holds. Together with the closedness of the stopping region, (\ref{28.11.2014.1}) implies that $\mathcal{S}$ is of the form given in (\ref{eq:boundary}) with 
boundary~$b(t):=\inf\{x\in\bbr_+\ |\ V(t,x)>G(t,x)\}$.

By (\ref{eq:osp_new}), for any $t\in[0,T)$ and $x\le y$, one has
{\allowdisplaybreaks{\begin{align}\label{eq:decreasing}
&\big(V(t,y)-G(t,y)\big)-\big(V(t,x)-G(t,x)\big)\notag\\
\geq&\mathbb{E}\left[\int_0^{\tau_{t,x}}F(t+u,X_{t+u}^{t,y})du\right]-\mathbb{E}\left[\int_0^{\tau_{t,x}}F(t+u,X_{t+u}^{t,x})du\right]\\
=&\mathbb{E}\left[\int_0^{\tau_{t,x}}F\left(t+u,\frac y x X_{t+u}^{t,x}\right)-F\left(t+u,X_{t+u}^{t,x}\right)du\right]\notag\\
\geq& 0 \notag,
\end{align}
}}where the last inequality holds by (\ref{fxgeq0}). By $V\ge G$, this implies (\ref{28.11.2014.1}).\\

{\em Step 2:} Let us now show that $t\mapsto b(t)$ in increasing. For $x\in\bbr_+$ and $s\leq t$, one has
{\allowdisplaybreaks{\begin{align}\label{ineq:UsUt}
 V(t,x)-G(t,x)=&e^{-r(1-\alpha)(t-s)}\sup_{\tau\in\mathcal{T}_{[0,T-t]}}\mathbb{E}\left[\int_0^\tau F(s+u,X_{t+u}^{t,x})du\right]\nonumber\\
       =&e^{-r(1-\alpha)(t-s)}\sup_\tau\mathbb{E}\left[\int_0^\tau F(s+u,X_{s+u}^{s,x})du\right]\\
    \leq&e^{-r(1-\alpha)(t-s)}\sup_{\tau\in\mathcal{T}_{[0,T-s]}}\mathbb{E}\left[\int_0^\tau F(s+u,X_{s+u}^{s,x})du\right],\nonumber\\
       =&e^{-r(1-\alpha)(t-s)}\left(V(s,x)-G(s,x)\right),\nonumber
\end{align}
}}
where the second supremum is taken over all $(\mathcal{F}_{s+u})_{u\in [0,T-t]}$--stopping times taking values in $[0,T-t]$. Since $V-G\ge 0$, (\ref{ineq:UsUt}) yields the implication
\begin{align}\label{imp:UsUt}
V(s,x)-G(s,x)=0\Rightarrow V(t,x)-G(t,x)=0,
\end{align}
which induces that $t\mapsto b(t)$ is increasing.\\

{\em Step 3:} Let us show that $b(t)>0$ for all $t\in[0,T)$. At first, suppose there exists $t^*\in(0,T)$ such that $b(t^*)=0$ (i.e., $t^*\not=0$). As $t\mapsto
b(t)$ is increasing, one has that $b(u)=0$ for all $u\in[0,t^*]$. 
As $X^{0,x}$ cannot reach $0$, we have $\tau_{0,x}\ge t^*$ for all $x>0$, which yields
\beao
& & \mathbb{E}\left[\int_0^{\tau_{0,x}}F(u,X_{u}^{0,x})du\right]\nonumber\\
& & \leq \mathbb{E}\left[\int_0^{t^*}F\left(u,X_{u}^{0,x}\right)1_{\{X_{u}^{0,x}\le f/2\}}du\right]
+\mathbb{E}\left[\int_0^{T} 0\vee F\left(u, X_u^{0,x}\right)1_{\{X_{u}^{0,x}>f/2\}}du\right]\nonumber\\
& & \leq t^* \mathbb{P}\left(\sup_{0\le u\le T} X_u^{0,x}\leq f/2\right)F(T,f/2) + Te^{rT}[\mu-r(1-\alpha)]\mathbb{E}\left[\sup_{0\le u\le T} X_u^{0,x}\right].
\eeao
For the first inequality, we use that $F(u,\wt{x})<0$ for $\wt{x}\le f/2$. For the second one, we use that $F$ is increasing in $x$ and for $x\le f$, increasing in $t$. 

We have that $F(T,f/2)<0$. In addition, $\mathbb{P}\left(\sup_{0\le u\le T} X_u^{0,x}\leq f/2\right)\to 1$ and 
$\mathbb{E}\left(\sup_{0\le u\le T} X_u^{0,x}\right)\to 0$ for $x\to 0$.  This yields that
\beao
\mathbb{E}\left[\int_0^{\tau_{0,x}}F(u,X_{u}^{0,x})du\right]<0\quad\mbox{for $x$ small enough,}
\eeao
which is a contradiction to the optimality of $\tau_{0,x}$. Thus $b(t^*)>0$ for all $t^*\in(0,T)$. $b(0)>0$ follows analogously by extending problem~(\ref{eq:osp}) 
to the interval~$[-1,T]$.\\

Thus, the theorem is now proven, besides the smooth-fit condition, the continuity of the exercise boundary and its terminal condition. These assertions need some more preparation
provided by the following lemmas.

\subsubsection{Continuity of the optimal stopping times}

The following two lemmas show that the optimal stopping
times are close together for stock price processes started in a
neighborhood of state and time.
\begin{lem}\label{lemma:approxtau}
 Let $\tau_{t,x}$ be defined as in (\ref{eq:ost}). 
 Fix $a>0$. Then, for all $\varepsilon>0$, there exists $\tilde \delta>0$ such that $\sup_{x\in[a,\infty)} \mathbb{P}(\left\{\tau_{t,x+\delta}-\tau_{t,x}>\varepsilon\right\})<\varepsilon$ for all $\delta\in(0,\tilde\delta)$.
\end{lem}
\begin{proof}
For all $\varepsilon,\delta>0$, one has
\begin{align*}
\mathbb{P}\left(\left\{\tau_{t,x+\delta}-\tau_{t,x}>\varepsilon\right\}\right)=\mathbb{P}\left(\left\{\tau_{t,x+\delta}-\tau_{t,x}>\varepsilon\right\}\cap
\left\{\tau_{t,x}\leq T-t-\varepsilon\right\}\right).
\end{align*}  
To compare the optimal stopping times for different state variables, we use that $X^{t,y}=yX^{t,1}$ for all $y\in\bbr_+$. Let $(B_u)_{u\geq0}$ be a standard 
$\mathbb{P}-$Brownian motion. One gets
\begin{align}\label{torewrite1}
&\mathbb{P}\left(\left\{\tau_{t,x+\delta}-\tau_{t,x}>\varepsilon\right\}\cap \left\{\tau_{t,x}\leq T-t-\varepsilon\right\}\right)\nonumber\\
\leq&\mathbb{P}\left(\left\{\min_{u\in[0,\varepsilon]}X^{t,x+\delta}_{t+\tau_{t,x}+u}>b(t+\tau_{t,x})\right\}\cap \left\{\tau_{t,x}\leq T-t-\varepsilon\right\}\right)\nonumber\\
\leq& \mathbb{P}\left(\left\{\min_{u\in[0,\varepsilon]} (x+\delta)X^{t,1}_{t+\tau_{t,x}+u}> xX^{t,1}_{t+\tau_{t,x}}\right\}\cap \left\{\tau_{t,x}\leq T-t-\varepsilon\right\}\right)\nonumber\\
\leq& P\left(\left\{\min_{u\in[0,\varepsilon]}\left\{\exp\left(\left(\mu-\frac{\sigma^2}2\right)u+\sigma B_u\right)\right\}>\frac{x}{x+\delta}\right\}\right)\nonumber\\
\leq&\Phi \left(\frac{-\ln\left(\frac{a}{a+\delta}\right)+(\mu-\sigma^2/2)\varepsilon}{\sigma\sqrt\varepsilon}\right)\nonumber\\
&-e^{2(\mu-\sigma^2/2)\ln\left(\frac{a}{a+\delta}\right)\sigma^{-2}}\Phi \left(\frac{\ln\left(\frac{a}{a+\delta}\right)+(\mu-\sigma^2/2)\varepsilon}{\sigma\sqrt\varepsilon}\right)\quad\forall x\in[a,\infty),
\end{align}
where the first inequality holds because $t\mapsto b(t)$ is increasing, the second one holds by  $x X^{t,1}_{t+\tau_{t,x}} \le b(t+\tau_{t,x})$, and the third one follows by the strong Markov property 
of $X^{t,1}$. For $\varepsilon>0$ fixed, the right-hand side of (\ref{torewrite1}) converges to $0$ for $\delta\downarrow 0$. \exit
\end{proof}
\begin{lem}\label{lemma:approxtau2}
Let $\tau_{t,x}$ be defined as in (\ref{eq:ost}).
Fix $a,b\in\mathbb{R}_+$ with $0<a<b$. Then, for all $\varepsilon>0$, there exists $\delta>0$ such that 
$\sup_{x\in[a,b]}\mathbb{P}(\left\{\left|\tau_{t_2,x}-\tau_{t_1,x}\right|>\varepsilon\right\})<\varepsilon$ for all $t_1,t_2$ with $0\leq t_1<t_2\leq T$ and $t_2-t_1<\delta$.
\end{lem}
\begin{proof}
Similarly to the proof of Lemma~\ref{lemma:approxtau}, for all $t_2\in(0,T]$, $\varepsilon>0$, and $u\in(0,\varepsilon)$, one has 
\begin{align*}
\mathbb{P}\left(\left\{\tau_{t_2-u,x}-\tau_{t_2,x}>\varepsilon\right\}\right)=\mathbb{P}\left(\left\{\tau_{t_2-u,x}-\tau_{t_2,x}>\varepsilon\right\}\cap
\left\{\tau_{t_2,x}\leq T-t_2+u-\varepsilon\right\}\right).
\end{align*}
To compare the optimal stopping times, we write $X^{t_2,x}$ in terms of the process $X^{t_1,x}$. Namely, by construction, one has 
\beam\label{28.12.2014.1}
X_{t_2+u}^{t_2,x}=x\frac{X^{t_1,x}_{t_2+u}}{X_{t_2}^{t_1,x}},\quad u\ge 0,
\eeam
where $t_2\ge t_1$.
In addition, the following argument uses the fact that $X_{t_2}^{t_1,x}\approx x$ for $t_2-t_1$ small.
Consequently, at the first time $X^{t_2,x}$ hits the boundary, the process~$X^{t_1,x}$ is not far away, and one can argue as in Lemma~\ref{lemma:approxtau}.  
Let $(B_u)_{u\geq0}$ be a standard $\mathbb{P}-$Brownian motion.  
Then, for $t_2-t_1\in(0,\varepsilon/2)$ and $\delta>0$, one gets
\begin{align}\label{tau_ungleichung_prep}
&\mathbb{P}\left(\left\{\tau_{t_1,x}-\tau_{t_2,x}>\varepsilon\right\}\cap \left\{\tau_{t_2,x}\leq T-t_1-\varepsilon\right\}\right)\nonumber\\
\leq & \mathbb{P}\left(\left\{\min_{u\in[0,\varepsilon-(t_2-t_1)]}X^{t_1,x}_{t_2+\tau_{t_2,x}+u} > b(t_2+\tau_{t_2,x})\right\}\cap \left\{\tau_{t_2,x}\leq T-t_1-\varepsilon\right\}\right)\nonumber\\
\leq& \mathbb{P}\left(\left\{\min_{u\in[0,\varepsilon-(t_2-t_1)]}\frac{X^{t_1,x}_{t_2+\tau_{t_2,x}+u}}{X^{t_1,x}_{t_2+\tau_{t_2,x}}} > \frac{x}{X^{t_1,x}_{t_2}}\right\}\cap \left\{\tau_{t_2,x}\leq T-t_1-\varepsilon\right\}\right)\nonumber\\
\leq & \mathbb{P}\left(\left\{\min_{u\in[0,\varepsilon/2]}\frac{X^{t_1,x}_{t_2+\tau_{t_2,x}+u}}{X^{t_1,x}_{t_2+\tau_{t_2,x}}} 
> \frac{x}{x+\delta}\right\}\cap \left\{\tau_{t_2,x}\leq T-t_1-\varepsilon\right\}\right)\nonumber\\
&+\mathbb{P}\left(\left\{X^{t_1,x}_{t_2}>x+\delta\right\}\right)\nonumber\\
\leq & \mathbb{P}\left(\left\{\min_{u\in[0,\varepsilon/2]}\left\{\exp\left(\left(\mu-\frac{\sigma^2}2\right)u+ \sigma B_u\right)\right\} > \frac{a}{a+\delta}\right\}\right)\nonumber\\
&+\mathbb{P}\left(\left\{{\exp\left(\left(\mu-\frac{\sigma^2}2\right)(t_2-t_1)+ \sigma B_{t_2-t_1}\right)>\frac{b+\delta}{b}}\right\}\right)\quad\forall x\in[a,b],
\end{align}
where the first inequality holds because $t\mapsto b(t)$ is increasing, the second one holds by (\ref{28.12.2014.1}) and $X^{t_2,x}_{t_2+\tau_{t_2,x}} \le b(t_2+\tau_{t_2,x})$, 
and for the fourth one, we use the strong Markov property of $X^{t_1,x}$.

As in Lemma \ref{lemma:approxtau}, one can choose $\delta>0$ small enough such that
\begin{align}\label{second_inequality}
\mathbb{P}\left(\min_{u\in[0,\varepsilon/2]}\left\{\exp\left(\left(\mu-\frac{\sigma^2}2\right)u+ \sigma B_u\right)\right\} >
\frac{a}{a+\delta}\right)<\frac\varepsilon2.
\end{align}
Because $B_{t_2-t_1}$ converges stochastically to $0$ for ${t_2-t_1}\downarrow0$, there exists $\delta>0$ such that
\begin{align}\label{deltamove}
\mathbb{P}\left(\left\{\exp\left(\left(\mu-\frac{\sigma^2}2\right)(t_2-t_1)+ \sigma B_{t_2-t_1}\right)>\frac{b+\delta}{b}\right\}\right)<\frac\varepsilon2
\end{align}
for all $t_1, t_2$ with $t_2-t_1<\delta$. So, it remains to show 
\begin{align}\label{backdirection}
\mathbb{P}\left(\tau_{t_2,x}-\tau_{t_1,x}>\varepsilon\right)<\varepsilon,\quad\mbox{for $t_2\ge t_1$ and $t_2-t_1$\ small enough}.
\end{align}
To estimate this probability, we renew $X$ at $t_1+\tau_{t_1,x}$, where it is sufficient to consider the set $\{t_1+\tau_{t_1,x}\ge t_2\}$. Then, we estimate $X^{t_1,x}_{t_2}$ from below
and not from above as in (\ref{tau_ungleichung_prep}). But, the calculations are completely analog to the estimation of $\mathbb{P}\left(\tau_{t_1,x}-\tau_{t_2,x}>\varepsilon\right)$, 
and we are done.
\exit
\end{proof}

\subsubsection{Smooth-fit condition}

Next, we show (\ref{eq:smoothfit}), i.e., the value function $V$ joints the payoff
function $G$ smoothly at the boundary. For $t\in[0,T)$ and $x=b(t)$, one has

\beam\label{xderivapprox:2}
 \frac{V(t,x+\varepsilon)-V(t,x)}{\varepsilon}\geq\frac{G(t,x+\varepsilon)-G(t,x)}{\varepsilon}& = & \partial_x G(t,x)\\
 & = & (1-\alpha)e^{r(1-\alpha)(T-t)}\nonumber
\eeam
for all $\eps>0$. On the other hand, one has
\beam\label{11.12.2014.1}
V(t,x+\varepsilon)-V(t,x) & \leq & \mathbb{E}\left[G\left(t+\tau_{t,x+\varepsilon},X_{t+\tau_{t,x+\varepsilon}}^{t,x+\varepsilon}\right)\right]
-\mathbb{E}\left[G\left(t+\tau_{t,x+\varepsilon},X_{t+\tau_{t,x+\varepsilon}}^{t,x}\right)\right]\nonumber\\
& = & \varepsilon \mathbb{E}\left[X^{t,1}_{t+\tau_{t,x+\varepsilon}}\partial_x G\left(t+\tau_{t,x+\varepsilon},x\right)\right]\nonumber\\
& \leq &\varepsilon \partial_x G(t,x)\mathbb{E}\left[X^{t,1}_{t+\tau_{t,x+\varepsilon}}\right],
\eeam
where the equality holds by the affine linearity of $G$ in $x$, and the second inequality holds as $t\mapsto \partial_x G(t,x)$ is decreasing.
By Lemma~\ref{lemma:approxtau} and uniform integrability, $\mathbb{E}\left[X^{t,1}_{t+\tau_{t,x+\varepsilon}}\right]$ converges to $1$ for $\eps\downarrow 0$. Thus, 
(\ref{xderivapprox:2}) and (\ref{11.12.2014.1}) establish the smooth-fit condition (\ref{eq:smoothfit}).  

\subsubsection{Continuity of the boundary}

\begin{prop}\label{prop:right-continuous}
The boundary $t\mapsto b(t)$ is right-continuous on $[0,T)$.
\end{prop}
\begin{proof}
Fix $t\in[0,T)$ and consider a sequence $t_n\downarrow t$ for $n\rightarrow\infty$. As $t\mapsto b(t)$ is increasing, $b(t+):=\lim_{s\downarrow t} b(s)$ exists. 
Since $(t_n,b(t_n))\in\mathcal{S}$ for all $n\geq1$, and $V$ and $G$ are continuous (Proposition \ref{prop:continue}), one gets that 
$V(t,b(t+))=G(t,b(t+))$, i.e., $(t,b(t+))\in\mathcal S$. This results $b(t+)\leq b(t)$. 
As $t\mapsto b(t)$ is increasing on $[0,T)$, the  claim is proved.\exit
\end{proof}
\begin{prop}\label{prop:left-continuous}
The boundary $t\mapsto b(t)$ is left-continuous on $(0,T)$.
\end{prop}

Firstly, note that by monotonicity, the limit $b(t-):=\lim_{s\uparrow t}b(s)$ exists and $b(t-)\le b(t)$. The proof of the proposition is divided into three steps. 
Under the assumption that a jump of the boundary occurs at some time $t$, in the first two steps, we find an upper bound for the $t$-derivative 
and the $x$-derivative of $V$ for points in time in a left neighborhood of  $t$ and prices between $b(t-)$ and $b(t)$.
Then, we use these upper bounds to argue with the PDE which is satisfied by $V$ in $\mathcal{C}$ to see that $V_{xx}>0$
is bounded away from~$0$. Roughly speaking, the contribution of $V_{xx}$ to the PDE is bounded away from zero by minus the drift rate~$F$ of the payoff function (cf. Step~3 of the
proof). In a neighborhood of the stopping region, this drift rate is strictly negative. 
Since $\{t\}\times[b(t-),b(t)]$ is part of the stopping region, where $V=G$, this turns out to 
be a contradiction to the linearity of $G$ in $x$ if $b(t)>b(t-)$.

This line of argument has already been applied to quite diverse payoff functions, see \cite{peskir}.

\begin{proof}
Suppose that the stopping boundary $b$ has a jump at $t$, i.e., $b(t)>b(t-)$.\\

\textit{Step 1 (upper bound for the $t-$derivative):} Let $\delta\in(0,t)$, $\varepsilon\in(0,t-\delta)$, and $x\in(b(t-\delta),b(t)]$. Define the stopping time
\beam\label{9.12.2014.1}
\sigma:=\inf\left\{u\ge 0\ |\ X^{t-\delta-\varepsilon,x}_{t-\delta-\eps+u}\le b(t-\delta+u)\right\}\wedge (T-(t-\delta))\in \mathcal{T}_{[0,T-(t-\delta)]},
\eeam
which applies the optimal stopping rule for the problem started in $t-\delta$ to the problem started in $t-\delta-\varepsilon$.
By construction, $\left(\sigma, X^{t-\delta-\eps,x}_{t-\delta-\eps+\sigma}\right)$ possesses the same distribution as 
$\left(\tau_{t-\delta,x}, X_{t-\delta+\tau_{t-\delta,x}}^{t-\delta,x}\right)$. 
Because $\sigma$ is in general sub-optimal for the problem started in $t-\delta-\eps$, one gets
{\allowdisplaybreaks{\begin{align}\label{t-derivative1}
&\frac{V(t-\delta,x)-V(t-\delta-\eps,x)}{\varepsilon}\nonumber\\
\leq&\frac{1}\varepsilon  \mathbb{E}\left[G(t-\delta+\tau_{t-\delta,x},X_{t-\delta+\tau_{t-\delta,x}}^{t-\delta,x})\right]
-\frac1\varepsilon\mathbb{E}\left[G(t-\delta-\eps+\sigma,X_{t-\delta-\eps+\sigma}^{t-\delta-\eps,x})\right]\nonumber\\
=& \mathbb{E}\left[\left((1-\alpha)X_{t-\delta+\tau_{t-\delta,x}}^{t-\delta,x} 
+\alpha P_0\right)e^{r(1-\alpha)(T-t+\delta-{\tau_{t-\delta,x}})}\right]\frac{\left(1-e^{r(1-\alpha)\varepsilon}\right)}{\varepsilon}.
\end{align}
}}
By the classic theory for parabolic equations, see, e.g., Friedman~\cite{Friedman} (Chapter~3) 
or Shiryaev~\cite{shiryaev2007optimal} (Theorem~15 in Chapter~3), one knows that  $V\in
\mathcal{C}^{1,2}$ in the continuation region. Therefore, $\partial_t V(t-\delta,x)$ exists for all $x>b(t-\delta)$ and
$$(V(t-\delta,x)-V(t-\delta-\eps,x))/{\varepsilon}\rightarrow
\partial_t V(t-\delta,x),\quad \eps\downarrow 0.$$ 
Together with (\ref{t-derivative1}), this implies
\begin{align*}
 \partial_t V(t-\delta,x)\leq -r(1-\alpha)\mathbb{E}\left[\left((1-\alpha) X_{t-\delta+\tau_{t-\delta,x}}^{t-\delta,x}
 +\alpha P_0\right)e^{r(1-\alpha)(T-(t-\delta)-{\tau_{t-\delta,x}})}\right]
\end{align*}

On the other hand, $(t,x)$ lies in the stopping region for all $x\le b(t)$ and thus $\tau_{t,x}=0$.
Since $b(t-)>0$, one can apply Lemma~\ref{lemma:approxtau2}, and $\tau_{t-\delta,x}\to 0$ in probability for $\delta\downarrow0$, where the convergence holds uniformly 
in $x\in[b(t-)/2,b(t)]$. By uniform integrability, one gets
\beam\label{t-deriv:approx}
& & \limsup_{\delta\downarrow0} \sup_{x\in (b(t-\delta),b(t)]}\left[\partial_t V(t-\delta,x) - \partial_t G(t,x)\right]\nonumber\\
& & \leq \sup_{x\in (b(t-),b(t)]}\left[-\left((1-\alpha)x+\alpha P_0\right)e^{r(1-\alpha)(T-t)}r(1-\alpha) -\partial_t G(t,x)\right]\nonumber\\
& & = 0
\eeam
\textit{Step 2 (upper bound for the $x-$derivative):} Let $\delta\in(0,t)$, $x\in(b(t-\delta),b(t)]$, and $\varepsilon\in(0,x-b(t-\delta))$. The arguments are similar to Step~1, 
but more convenient to write down because $\tau_{t-\delta,x}$ is already an admissible stopping time for the problem started 
in $(t-\delta,x-\eps)$ and need not be transformed as in (\ref{9.12.2014.1}).
One gets
\begin{align}\label{x-derivative1}
&\frac{V(t-\delta,x)-V(t-\delta,x-\eps)}{\varepsilon}\nonumber\\
\leq& \frac{1}\varepsilon \left[\mathbb{E}\left[G\left(t-\delta+\tau_{t-\delta,x},X_{t-\delta+\tau_{t-\delta,x}}^{t-\delta,x}\right)\right]
-\mathbb{E}\left[G\left(t-\delta+\tau_{t-\delta,x},X_{t-\delta+\tau_{t-\delta,x}}^{t-\delta,x-\eps}\right)\right]\right]\nonumber\\
=&(1-\alpha)
\mathbb{E}\left[X_{t-\delta+\tau_{t-\delta,x}}^{t-\delta,1}e^{r(1-\alpha)(T-(t-\delta)-\tau_{t-\delta,x})}\right],
\end{align}
Again, by  $V\in \mathcal{C}^{1,2}$ in the continuation region, one gets that $\partial_x V(t-\delta,x)$ exists for all $x>b(t-\delta)$ and
$(V(t-\delta,x)-V(t-\delta-\eps,x))/{\varepsilon}\rightarrow
\partial_x V(t-\delta,x)$ for $\varepsilon\downarrow0$. 
Together with (\ref{x-derivative1}), one gets
\begin{align*}
\partial_x V(t-\delta,x)\leq (1-\alpha)\mathbb{E}\left[X_{t-\delta+\tau_{t-\delta,x}}^{t-\delta,1}e^{r(1-\alpha)(T-(t-\delta)-\tau_{t-\delta,x})}\right].
\end{align*}
Again, by $\tau_{t,x}=0$ and $b(t-)>0$, an application of Lemma~\ref{lemma:approxtau2} yields  
\begin{align}\label{x-deriv:approx}
\limsup_{\delta\downarrow0}\sup_{x\in(b(t-\delta),b(t)]}  \partial_x V(t-\delta,x)\leq (1-\alpha)e^{r(1-\alpha)(T-t)}.
\end{align}
The RHS is $\partial_x G(t,x)$ that does not depend on $x$.\\

\textit{Step 3 (Conclusion of the left-continuity):} Now, we want to lead the assumption $b(t)>b(t-)$ to a contradiction. 

Let $x^*:=(b(t-)+b(t))/2$. By Remark~\ref{rem:conti_region}, one has $b(t)\le f$ and thus, by (\ref{fgeq0}), it follows that $F(t,x^*)<0$. 
By Step~1 and Step~2, there exists $\delta>0$ such that for all $s\in[t-\delta,t)$ and $x\in(b(s),x^*]$,
one has that
\beao
\mu x \partial_x V(s,x) + \partial_t V(s,x) & \le & \mu x \partial_x G(s,x) + \partial_t G(s,x) -\frac{F(t,x^*)}3\\ 
& = & F(s,x) - \frac{F(t,x^*)}3\\
& \le & \frac23 F(t,x^*) - \frac{F(t,x^*)}3\\
& = & \frac{F(t,x^*)}3,
\eeao
where the second inequality holds by the continuity of $F$ and its monotonicity in $x$. Again, by Theorem~15 in Chapter~3 of \cite{shiryaev2007optimal},
we know that the value function $V$ solves the PDE
\begin{equation*}
\partial_t V + \mu x \partial_x V + \frac{\sigma^2}{2}x^2\partial_{xx} V = 0
\end{equation*} 
in the continuation region $\mathcal{C}$. Thus, we have
\beao 
\partial_{xx} V(s,x)\ge -\frac{2F(t,x^*)}{3\sigma^2 x^2} \ge -\frac{2F(t,x^*)}{3\sigma^2 {x^*}^2} =: C >0,\quad \forall s\in[t-\delta,t),\ x\in(b(s),x^*]. 
\eeao
By $V(s,b(s))=G(s,b(s))$, $\partial_x V(s,b(s))=\partial_x G(s,b(s))$ (smooth-fit condition), $\partial_{xx}G=0$, and the Newton-Leibniz
formula, it follows that 
\begin{align}\label{xx-derivative}
 V(s,x^*)-G(s,x^*)=\int_{b(s)}^x\int_{b(s)}^u \partial_{xx}V(s,v)-\partial_{xx}G(s,v)dvdu \ge \frac{C(x^*-b(s))^2}2
\end{align}
As this holds for all $s\in[t-\delta,t)$ and $V-G$ is continuous, one concludes that 
\beao
V(t,x^*)-G(t,x^*) \ge \frac{C(x^*-b(t-))^2}2>0,
\eeao
which is a contradiction to the fact that $(t,x^*)$ lies in the stopping region. So, it can be concluded that $b(t-)=b(t)$,
and the continuity of the boundary is established.
\exit
\end{proof}

\subsubsection*{Terminal Condition of the boundary}

By Remark~\ref{rem:conti_region}, $b(T-)=\lim_{t\uparrow T}b(t)$ cannot exceed the boundary~$f$, above which
the drift rate of the payoff process is positive. It remains to exclude that $b(T-)<f$. But, this is done with the same arguments as in the proof of the
left-continuity of the boundary, using the fact that $V(T,x)=G(T,x)$ for all $x\in [b(T-),f].$

\subsection{Proof of Proposition \ref{2.12.2014.1}}

(i) Let $0\le \sigma_1\le\sigma_2$ and w.l.o.g. $t=0$. For two independent standard Brownian motions $B$ and $\wt{B}$, the process
\beao
X_s=x\exp\left(\mu s + \sigma_1 B_s - \frac{\sigma_1^2}2 s + \sqrt{\sigma_2^2 - \sigma_1^2} \wt{B}_s - \frac{\sigma_2^2 - \sigma_1^2}2 s\right),\quad s\ge 0,
\eeao
possesses the same law as the stock price from (\ref{eq:stock}) with $\sigma=\sigma_2$. In addition, $X$ is Markovian w.r.t. the filtration $(\mathcal{F}^{B,\wt{B}}_s)_{s\in[0,T]}$
which is generated by $B$ and $\wt{B}$.
This implies that $V$ from (\ref{eq:osp_at0}) with standard deviation $\sigma_2$ coincides with the value of the problem
\beam\label{2.12.2014.2}
\sup_\tau \mathbb{E}\left[G(\tau,X_{\tau})\right],
\eeam
where the supremum is taken over all $(\mathcal{F}^{B,\wt{B}}_s)_{s\in[0,T]}$--stopping times~$\tau$. 
Now consider the artificial optimal stopping problem where the second Brownian motion $\wt{B}$ that enters into the stock price is not 
observable to the maximizer. This corresponds to the restriction to $(\mathcal{F}^B_s)_{s\in[0,T]}$--stopping times. Of course, the latter supremum is at least 
as high as the previous one. On the other hand, for any $(\mathcal{F}^B_s)_{s\in[0,T]}$--stopping time~$\tau$, we have
\beao
& & \mathbb{E}\left[G(\tau,X_{\tau})\right]\\ 
& & = \mathbb{E}\Big[\left((1-\alpha)x\exp\left(\mu \tau + \sigma_1 B_\tau - \frac{\sigma_1^2}2 \tau + \sqrt{\sigma_2^2 - \sigma_1^2} \wt{B}_\tau - \frac{(\sigma_2^2 
- \sigma_1^2)\tau}2\right) + \alpha P_0\right)\\ 
& & \quad \times e^{r(1-\alpha)(T-\tau)}\Big]\\
& & = (1-\alpha)\mathbb{E}\Big[e^{r(1-\alpha)(T-\tau)} 
x\exp\left(\mu \tau + \sigma_1 B_\tau - \frac{\sigma_1^2}2 \tau\right)\\
& & \qquad\qquad\qquad \times  E\left(\exp\left(\sqrt{\sigma_2^2 - \sigma_1^2} \wt{B}_\tau - \frac{\sigma_2^2 - \sigma_1^2}2\tau\right)\ |\ \mathcal{F}^B_T\right)\Big]\\
& & + \alpha P_0 \mathbb{E}\left[e^{r(1-\alpha)(T-\tau)}\right]\\
& & = \mathbb{E}\left[(1-\alpha)e^{r(1-\alpha)(T-\tau)} 
x\exp\left(\mu \tau + \sigma_1 B_\tau - \frac{\sigma_1^2}2 \tau\right) + \alpha P_0 e^{r(1-\alpha)(T-\tau)}\right],
\eeao 
where the conditional expectation is $1$ because $\tau$ is $\mathcal{F}^B_T$--measurable whereas $\wt{B}$ is independent of $\mathcal{F}^B_T$. 
It follows that the value of problem~(\ref{2.12.2014.2}) restricted to all $(\mathcal{F}^B_s)_{s\in[0,T]}$--stopping times coincides with $V$ from (\ref{eq:osp_at0}) 
with smaller standard deviation $\sigma_1$. Thus, one has $V^{\sigma_1}\le V^{\sigma_2}$. Note that the only property of the payoff function~$G$ we use is its affine linearity 
in $X_\tau$.\\ 
 
(ii) For $\sigma=0$, the optimal stopping problem (\ref{eq:osp}) reads
$$V(t,x)=\sup_{u\in[0,T-t]}G(t+u,xe^{\mu u})=\sup_{u\in[0,T-t]}\left((1-\alpha)xe^{\mu u}+\alpha P_0\right)e^{r(1-\alpha)(T-t-u)}.$$
By simple algebra, one calculates that $d^2/du^2 G(t+u,xe^{\mu u})>0$. So, we conclude that for fixed $x\in\mathbb{R}_+$, the maximum of $u\mapsto G(t+u,xe^{\mu u})$ 
is either attained at $u=0$ or $u=T-t$. For $b(t)$ given by (\ref{17.12.2014.1}), one has $G(t,b(t))= G(T,b(t) e^{\mu (T-t)})$, i.e., the investor is indifferent between stopping at
$t$ and $T$. It implies that $(t,b(t))$ lies 
in the stopping region. On the other hand, by $\partial_x G(t,x)=(1-\alpha)e^{(1-\alpha)r(T-t)} < e^{\mu(T-t)}\partial_x G(T,x)$ for all $x\in\bbr_+$, one has that $(t,x)$ lies 
in the continuation region for all $x>b(t)$. This implies that $b(t)$ is indeed
the optimal exercise boundary from (\ref{eq:boundary}).

It remains to show that $\tau_{t,x}=T-t$ for all $x>b(t)$. As $\sup_{u\in[0,T-t]}G(t+u,xe^{\mu u})$ is not attained at any $u\in(0,T-t)$, $X^{t,x}$ cannot hit the boundary before $T$
if it starts above it.


\begin{thebibliography}{21}

\bibitem{Avellaneda} Avellaneda, M. and Lipkin, M. D. (2003) A market-induced mechanism
for stock pinning. \textit{Quantitative Finance}, 3(6), 417--425.

\bibitem{Baurdoux} {Baurdoux, E. J., Chen, N, Surya, B. A. and Yamazaki, K. (2014). Double optimal stopping of a Brownian bridge.
\textit{Submitted for publication}, http://arxiv.org/abs/1409.2226}.

\bibitem{bentaharsonertouzi07} Ben Tahar, I., Soner, M., and Touzi, N. (2007) The dynamic programming equation for the problem of optimal investment 
under capital gains taxes, \textit{SIAM Journal on Control and Optimization}, 46(5), 1779-1801.

\bibitem{bentaharsonertouzi10} Ben Tahar, I.~B., Soner, H.~M., and Touzi, N. (2010) Merton Problem with taxes: characterisation, computation and approximation, \textit{SIAM Journal 
on Financial Mathematics}, 1(1), 366-395.

\bibitem{buescu2007note} Buescu, C., Cadenillas, A., and Pliska, S.~R. (2007). A note on the effects of taxes on optimal investment,
\textit{Mathematical Finance}, 17(4), 477-485.

\bibitem{cadenillas1999optimal} Cadenillas, A. and Pliska, S.~R. (1999). Optimal trading of a security when there are taxes and transaction costs,
\textit{Finance and Stochastics}, 3(2), 137--165.

\bibitem{constantinides.1983} Constantinides, G.~M. (1983). Capital market equilibrium with personal taxes,
\textit{Econometrica}, 51(3), 611-636.

\bibitem{Dai} Dai, M. and Zhong, Y.F. (2012). Optimal stock selling/buying strategy with reference to the ultimate average,
\textit{Mathematical Finance}, 22(1), 165-184.

\bibitem{dem1} DeMiguel, V. and Uppal, R. (2005). Portfolio investment with the exact tax basis via nonlinear programming,
\textit{Management Science}, 51(2), 277-290.

\bibitem{duToit} Du Toit, J. and Peskir, G. (2009). Selling a stock
at the ultimate maximum. \textit{Annals of Applied Probability}, 19(3), 983--1014.

\bibitem{dyb1} Dybvig, P. and Koo, H. (1996) Investment with Taxes. \textit{Washington University, St. Louis, MO, Working Paper}.
    
\bibitem{Friedman} Friedman, A. (1964). \textit{Partial Differential Equations of Parabolic
Type}, Prentice-Hall, Florida.


\bibitem{jouini1} Jouini, E., Koehl, P.-F., and Touzi, N.~(2000).
Optimal investment with taxes: an existence result.
\textit{Journal of Mathematical Economics}, 33(4), 373-388.

\bibitem{jouini2} Jouini, E., Koehl, P.-F., and Touzi, N.~(1999). Optimal investment with taxes: an optimal control problem with endogenous delay.
\textit{Nonlinear Analysis}, 37, 31--56.

\bibitem{kuehn.ulbricht} K\"uhn, C. and Ulbricht, B. (2014), Modeling capital gains taxes for trading strategies of infinite variation.
\textit{Submitted for publication}. http://arxiv.org/abs/1309.7368 





\bibitem{peskir} Peskir, G. and Shiryaev, A.N. (2006). Optimal Stopping and Free Boundary Problems, \textit{Birkh\"auser}.

\bibitem{seifried.2010} Seifried, F. (2010) Optimal Investment with Deferred Capital Gains Taxes, \textit{Mathematical Methods of Operations Research}, 71(1), 181-199.

\bibitem{Shiryaev} Shiryaev, A.N., Xu, Z.Q., and Zhou, X.Y. (2008). Thou shalt buy and hold. \textit{Quantitative Finance}, 8(8), 765--776.

\bibitem{shiryaev2007optimal} Shiryaev, A. N. (2008). \textit{Optimal stopping rules} (Vol. 8). Springer.

\end{thebibliography}
\end{document}